\documentclass[journal,onecolumn]{IEEEtran}
\usepackage[cmex10]{amsmath}
\usepackage{amssymb,amsthm}
\usepackage{url}
\usepackage{verbatim}
\usepackage{color}
\usepackage{bbm}

\newtheorem{theorem}{Theorem}

\newtheorem{lemma}[theorem]{Lemma}

\newtheorem{proposition}[theorem]{Proposition}

\newcommand{\mc}[1]{\mathcal{#1}}
\newcommand{\bra}[1]{\langle #1 |}
\newcommand{\ket}[1]{| #1 \rangle}
\newcommand{\op}[2]{|#1\rangle \langle #2|}
\newcommand{\ip}[2]{\langle #1|#2\rangle}
\newcommand{\proj}[1]{| #1 \rangle\!\langle #1 |}

\DeclareMathOperator{\tr}{Tr}
\DeclareMathOperator{\Sch}{Sch}

\begin{document}
\title{When do Local Operations and Classical Communication Suffice for Two-Qubit State Discrimination?}

\author{Eric~Chitambar,~Runyao~Duan,~Min-Hsiu~Hsieh
  \thanks{Eric Chitambar is with the Department of Physics and Astronomy at Southern Illinois University, Carbondale.  Runyao Duan and Min-Hsiu Hsieh are with Centre for Quantum Computation \& Intelligent Systems, Faculty of Engineering and Information Technology, University of Technology, Sydney, P.O. Box 123, Broadway NSW 2007, Australia.
  
This paper was accepted as a long talk at Asian Quantum Information Science (AQIS) conference, 2013.}}

\maketitle

\begin{abstract}
In this paper we consider the conditions under which a given ensemble of two-qubit states can be optimally distinguished by local operations and classical communication (LOCC).  We begin by completing the \emph{perfect} distinguishability problem of two-qubit ensembles - both for separable operations and LOCC - by providing necessary and sufficient conditions for the perfect discrimination of one pure and one mixed state.  Then for the well-known task of minimum error discrimination, it is shown that \textit{almost all} two-qubit ensembles consisting of three pure states cannot be optimally discriminated using LOCC.  This is surprising considering that \textit{any} two pure states can be distinguished optimally by LOCC.  Special attention is given to ensembles that lack entanglement, and we prove an easy sufficient condition for when a set of three product states cannot be optimally distinguished by LOCC, thus providing new examples of the phenomenon known as ``non-locality without entanglement.''  We next consider an example of $N$ parties who each share the same state but who are ignorant of its identity.  The state is drawn from the rotationally invariant ``trine ensemble,'' and we establish a tight connection between the $N$-copy ensemble and Shor's ``lifted'' single-copy ensemble.  For any finite $N$, we prove that optimal identification of the states cannot be achieved by LOCC; however as $N\to\infty$, LOCC can indeed discriminate the states optimally.  This is the first result of its kind.  Finally, we turn to the task of unambiguous discrimination and derive new lower bounds on the LOCC inconclusive probability for symmetric states. When applied to the double trine ensemble, this leads to a rather different distinguishability character than when the minimum-error probability is considered.
\end{abstract}


\section{Introduction}

The ability to distinguish one physical configuration from another lies at the heart of information theory.  When quantum systems are used for information transmission, messages are encoded into quantum states, and the processing of this information in a faithful manner requires the encoded states to be distinguishable from one another.  Hence, a fundamental topic in quantum information is the problem of \textit{state discrimination}, which investigates how well ensembles of quantum states can be distinguished under various physical conditions.

One important operational setting in which questions of distinguishability emerge is the so-called ``distant lab'' scenario.  Here, some multiparty quantum state is distributed to spatially separated quantum labs, and the various parties use local measurements combined with classical communication to try and identify their state.  This operational setting is also known as LOCC (Local Operations and Classical Communication), and the study of LOCC operations has played an important role in developing our understanding of not only quantum information processing, but also the nature of quantum entanglement itself.  For instance, as demonstrated by the fundamental task of quantum teleportation \cite{Bennett-1993a}, viewing quantum communication in the LOCC setting allows us to cleanly separate entangled bits (ebits), qubits, and classical bits as distinct resources that can be used for transmitting information between different parties.  Furthermore, the celebrated tasks of quantum key distribution \cite{Bennett-1984a} and entanglement distillation \cite{Bennett-1996a} are all procedures performed within the LOCC paradigm.  Yet at the same time, LOCC operations can be viewed as a more basic concept than quantum entanglement since a multipartite quantum state possesses entanglement if and only if this state cannot be generated by LOCC operations \cite{Werner-1989a, Popescu-1997a, Plenio-2007a}.

As LOCC operations are just a subset of all possible physical operations, certain state discrimination tasks become impossible when the distant-lab constraint is imposed.  For instance, it is well-known that a set of quantum states can be perfectly distinguished if and only if the states are orthogonal.  For multi-party states, this statement is still true; however, the measurement used to discriminate the states may need to be a \textit{global} measurement that coherently acts across all the subsystems.  In many cases, this global measurement cannot be implemented locally, thus making the states indistinguishable by LOCC unless some identification error occurs (examples can be found in Refs. \cite{Ghosh-2001a, Ghosh-2004a, Walgate-2002a, Horodecki-2003a, Nathanson-2005a,  Watrous-2005b, Feng-2009a, Yu-2010a, Duan-2010a, Cosentino-2013a}).  Such a limitation allows for the implementation of important information-theoretic objectives such as data hiding \cite{Terhal-2001a, DiVincenzo-2002a} and secret sharing \cite{Hillery-1999a, Gottesman-2000a}.  For more general sets of states (possibly non-orthogonal), one can quantify their distinguishability using a variety of different measures, and in this paper, we consider both the \textit{minimum error} guessing probability and the \textit{maximum conclusive} (or unambiguous) probability for a given ensemble.  Both of these figures of merit are given in terms of some success probability that has been optimized over all possible measurements.  When an LOCC measurement can obtain the same success probability as the global optimal, then we say that LOCC is able to optimally distinguish the ensemble with respect to the particular figure of merit, otherwise it cannot.  The underlying question studied in this paper is when it's possible for LOCC to perform optimal state discrimination.  

In general, this question is quite difficult due to the complexity of LOCC: the global communication among the parties enables the choice of local measurement by one party at one particular round to depend on the measurement outcomes of all the other parties in previous rounds.  It is often helpful to visualize a general LOCC operation as a tree where each node indicates a particular choice of local measurement and each branch corresponds to a particular sequence of measurement outcomes.  Deciding whether or not a certain discrimination task is feasible by LOCC therefore amounts to a consideration of all such possible trees.

Despite its complexity, partial progress has been made in understanding conditions in which LOCC can perform optimal state discrimination.  Most notably is the discovery that \textit{any} two orthogonal pure states can be perfectly distinguished using LOCC \cite{Walgate-2000a}.  A similar result holds for pairs of non-orthogonal states in which again, LOCC can obtain the optimal discrimination success probability that is physically possible; this is true for both minimum error discrimination \cite{Virmani-2001a} and optimal conclusive discrimination \cite{Chen-2002a, Ji-2005b}.  This finding is particularly relevant to the current paper since we will show that, in sharp contrast, almost all triples of two-qubit states \textit{cannot} be optimally distinguished by LOCC.  

The fact that non-LOCC measurements can distinguish certain ensembles better than any LOCC strategy may not be overly surprising when the ensemble states possess entanglement.  This is because entanglement embodies some non-local property of two or more systems, and thus a global measurement across all systems is needed in general to discriminate among entangled states (this is essentially at the heart of superdense coding \cite{Bennett-1992b}).  However, rather surprisingly, certain ensembles exist consisting of unentangled states that cannot be distinguished optimally using LOCC \cite{Bennett-1999a}.  This phenomenon is often called ``nonlocality without entanglement,'' and it essentially reflects that fact that nonlocality and entanglement are two different physical properties of multipartite quantum systems.  Understanding the difference between the two is an important problem in quantum information science, and thus a main objective of this paper is to study, in particular, LOCC discrimination of ensembles that lack entanglement.

\subsubsection*{Summary of Results}

The body of this paper begins in Section \ref{Sect:2x2Perfect} with a return to the problem of perfect state discrimination among two-qubit orthogonal states.  While our primary interest is LOCC discrimination, we will also consider discrimination by the more general class of separable operations (SEP).  This problem has been solved for almost all types of ensembles, and we solve the missing piece of perfect discrimination between one pure state and one mixed state.  Interestingly, we find that SEP is more powerful than LOCC in the sense that certain two-state ensembles are distinguishable by SEP but not LOCC.  This result allows us to later construct in Section \ref{Sect:PureProd} examples of one pure product state and one (non-orthogonal) separable mixed state that cannot be optimally distinguished by LOCC.  Thus, we obtain a large class of two-state ensembles which demonstrate nonlocality without entanglement.

Section \ref{Sect:MinError} investigates the problem of minimum-error discrimination between linearly independent states.  However, we prove that this seemingly more general problem actually reduces to the problem of perfect discrimination of orthogonal states.  This reduction therefore allows us to apply the results of Section \ref{Sect:Main} toward the problem of minimum-error discrimination of non-orthogonal (linearly independent) states.  As a result, we obtain in \ref{Sect:AlmostAll} our main result that almost any three states cannot be optimally distinguished by LOCC.  More precisely, if we select a three-state ensemble by randomly choosing our states, then almost surely will LOCC fail to discriminate them as successfully as a more general global measurement.  Sections \ref{Sect:3states}--\ref{Sect:Ncopy} then restrict attention to ensembles composed of unentangled states.  We are able to obtain a simple necessary condition for when three product states cannot be distinguished optimally by LOCC.  With this result, new examples of nonlocality without entanglement can easily be constructed.

In Section \ref{Sect:Ncopy}, we move beyond two-qubit ensembles and consider the optimal discrimination of three symmetric $N$-qubit states.  The specific ensemble we analyze is the $N$-copy generalization of the celebrated double trine ensemble \cite{Peres-1991a}.  We prove that for any finite $N$, the ensemble cannot be optimally discriminated using $N$-party LOCC.  However as $N\to\infty$, we give a protocol that indeed achieves optimal (perfect) discrimination.  This is quite different from the $N$-copy discrimination among two possible pure states which can always be accomplished optimally by LOCC \cite{Acin-2005a}.

Finally, in Section \ref{Sect:Unambig}, we consider the task of unambiguous discrimination by LOCC.  We derive new upper bounds on the LOCC success probability for ensembles of two-qubit symmetric states.  With this, simple examples can be found when LOCC is insufficient for optimal unambiguous discrimination.  We again consider the double trine and find that surprisingly, separable operations and LOCC operations perform equally well in the task of unambiguous discrimination, a completely different behaviour than when the figure of merit is the minimum error probability \cite{Chitambar-2013b}.   Before getting to all of these results, we first review some basic definitions and describe essential concepts for our investigation.

\section{Definitions and Notation} 

\label{Sect:Defn}

Let $\mc{H}=\mc{H}_1\otimes...\otimes\mc{H}_N$ denote the underlying Hilbert space for an $N$-partite quantum system.  Here, $\mc{H}_{k}$ is the local system of party $k$ having dimension $d_k$ so that $d=\prod_{i=1}^Nd_i$ is the total dimension.  The set of bounded linear operators acting on $\mathcal{H}$ will be denoted by $\mathcal{B}(\mc{H})$ and $\mathbb{I}_d$ is the identity element in $\mathcal{B}(\mc{H})$.  The (one-copy) \textit{quantum state discrimination problem} involves the task of correctly identifying a quantum state that is randomly sampled from an ensemble $\mathcal{E}=\{\rho_i,p_i\}_{i=1}^n$, where $\rho_i\in\mathcal{B}(\mc{H})$ and $p_i$ is the probability of obtaining $\rho_i$.  The ``which state'' classical information is extracted from the sampled state using a positive operator-valued measure (POVM), which is a collection of positive semidefinite operators $\Pi=\{\Pi_i\}_{i=1}^n$ acting on $\mc{B}(\mc{H})$ such that $\sum_{i=1}^n\Pi_i=\mathbb{I}_d$.  The total identification success probability of the POVM $\Pi$ is $\Pi(\mathcal{E}):=\sum_{i=1}^np_i \tr[\Pi_i\rho_i]$, and the \textit{minimum error probability} is given by
\begin{equation}
P_{err}(\mathcal{E})=\min_{\Pi}\quad (1- \Pi(\mathcal{E})).
\end{equation}
Here the minimization is taken over all $n$-outcome POVMs, and a minimum can indeed be obtained since the set of POVMs is compact.  

For the task of unambiguous discrimination, an extra outcome $\Pi_0$ is appended to the set of POVMs, and an additional constraint must be satisfied that $\tr[\Pi_i\rho_j]=0$ whenever $i\not=j$.  Under this condition, no error will ever be made when guessing the state; however, the outcome ``0'' represents an inconclusive outcome and no guess is made on the state's identity.  The \textit{minimum inconclusive probability} is thus given by the following
\begin{align}
P_{inc}(\mathcal{E})=\min_{\Pi} \quad& \sum_{i=1}^n \tr[\Pi_0\rho_i]\notag\\
s.t.\quad& \tr[\Pi_i\rho_j]=0\quad i\not=j>0.
\end{align}
This time, the minimization is taken over all $(n+1)$-outcome POVMs.

We say the POVM $\Pi$ is \emph{separable} (SEP) if for each $i$, $\frac{\Pi_i}{\tr(\Pi_i)}$ can be expressed as a convex sum of product projectors $\op{a_1}{a_1}\otimes...\otimes\op{a_N}{a_N}$, where $\op{a_k}{a_k}\in\mc{B}(\mc{H}_k)$.  The main interest in studying separable POVMs is that any LOCC POVM is a separable POVM \cite{Bennett-1999a}, and therefore SEP offers a useful approximation to LOCC.  However, there exists non-LOCC operations that are nevertheless separable.  Since the POVM elements of SEP contain no entanglement, studying these non-LOCC separable operations provide one way to understand the subtle difference between entanglement and nonlocality.

A general LOCC POVM is very complex and fortunately we will not need a precise characterization of them \cite{Chitambar-2012c}.  Roughly speaking an LOCC protocol consists of successive local measurements which can each be described by a set of Kraus operators $\{M_\lambda^{(k)}\otimes\mathbb{I}^{(\overline{k})}\}_\lambda$ where $M_\lambda^{(k)}\in\mc{B}(\mc{H}_k)$ and $\sum_\lambda (M_\lambda^{(k)})^\dagger M_\lambda^{(k)}=\mathbb{I}_k$.  The notation reflects that party $k$ performs a measurement while the other parties act trivially.  The measurement outcome $\lambda$ is announced to all the parties and some other party chooses a local measurement to perform based on the information $\lambda$; this process continues round after round.  For the state discrimination problem, the parties assign each state to a collection of possible measurement outcome sequences, and they guess the state's identity based on this assignment and what they happen to measure.

For most of this paper, we will only consider two-qubit systems.  In such small dimensions, problems of distinguishability become tractable, yet there is still enough degrees of freedom for interesting phenomenon to emerge.  A general $2\otimes 2$ pure state $\ket{\psi}$ can be uniquely represented by the $2\times 2$ matrix $\psi$ given by $\ket{\psi}=\mathbb{I}\otimes\psi\ket{\Psi^+}$, where $\ket{\Psi^+}$ is the maximally entangled state $\ket{\Psi^+}=\sqrt{1/2}(\ket{00}+\ket{11})$.  One measure of the entanglement possessed by $\ket{\psi}$ is its \textit{concurrence}, which is defined by $C(\psi)=|\det(\psi)|$ \cite{Wootters-1998a}.

The entanglement of $2\otimes 2$ states has an additional feature of being completely detectable by the positive partial transpose (PPT) criterion.  For a bipartite matrix $M$, we let $M^\Gamma$ denote its partial transpose, which is defined by performing the transpose operation on Bob's system alone with respect to some fixed basis.  A two-qubit density matrix is separable (i.e. does not possess entanglement) if and only if its partial transpose has no negative eigenvalues \cite{Horodecki-1996a}. 

A useful fact about two-qubit ensembles is that for any three orthogonal states, there exists a unique state orthogonal to all three.  Thus, for a given state $\ket{\Phi}$, we will let $\{\ket{\Phi}\}^\perp$ denote the orthogonal complement of $\ket{\Phi}$.  Another important property of two-qubit spaces is given by the first of the following two technical lemmas.  Both of these will be used heavily in Section \ref{Sect:2x2Perfect}, and they each refer to ``antiparallel'' eigenvalues, which are any two complex eigenvalues $z_1$ and $z_2$ related by $z_1=az_2$ for some negative real number $a$.
\begin{lemma}
\label{Lem:MixedSep}
Suppose $S$ is a two-dimensional subspace such that the subspace projector $P_S$ is separable.  Then $S$ has an orthonormal product basis.  Furthermore, if $\ket{\psi}$ and $\ket{\Phi}$ are any two entangled orthogonal states in $S^\perp$, then (i) $\psi\Phi^{-1}$ have two anti-parallel eigenvalues, and  (ii) $C(\psi)=C(\Phi)$.
\end{lemma}
\begin{proof}
As $P_S$ is separable of rank two, it can be written as $P_S=x\proj{ab}+y\proj{cd}$, where $x,y$ are positive numbers and
$\ket{ab},\ket{cd}$ are two linearly independent product states spanning $S$ \cite{Sanpera-1998a}.  This linear independence ensures that the operators \[\{\op{ab}{ab},\op{ab}{cd},\op{cd}{ab},\op{cd}{cd}\}\] are also linearly independent.  Therefore, the condition $P_S^2=P_S$ implies that $\ket{ab}$ and $\ket{cd}$ are orthonormal.  To prove the final assertion, note that when $S$ has an orthonormal product basis, so does $S^\perp$.  Either $S^\perp$ only contains product states, in which case the lemma is trivially satisfied, or any two entangled orthogonal states in $S^\perp$ will take the form $\ket{\psi}=\cos\theta\ket{01}+e^{i\varphi}\sin\theta\ket{10}$ and $\ket{\Phi}=-\sin\theta\ket{01}+e^{i\varphi}\cos\theta\ket{10}$.  It can immediately be seen that these states have the same concurrence.  Finally, a simple calculation reveals that the eigenvalues of $\psi\Phi^{-1}$ are $-\tan\theta$ and $\cot\theta$, which satisfies part (i) of the lemma. 
\end{proof}
\begin{lemma}[Duan \textit{et al.} \cite{Duan-2007a}]
\label{Lem:Runyao}
Let $\ket{\psi}$ and $\ket{\Phi}$ be bipartite entangled pure states with
$\lambda \geq 0$. Then $\rho(\lambda)=\proj{\psi}+\lambda\proj{\Phi}$
is separable if and only if (i) $\psi\Phi^{-1}$ has two antiparallel
eigenvalues and (ii) $\lambda=C(\psi)/C(\Phi)$.
\end{lemma}
\begin{proof}
See Ref. \cite{Duan-2007a}.
\end{proof}

\section{Completing the $2\otimes 2$ Perfect Discrimination Picture}

\label{Sect:2x2Perfect}

We begin by providing a full solution to the perfect distinguishability problem of two-qubit ensembles, both for separable operations and LOCC.  Note that this problem has been heavily studied, so much of our work here will be a recollection of previously known results.  However, we do provide simpler necessary and sufficient conditions for some occasions that are easier to use than those previous results. As perfect distinguishability requires orthogonality of the states, there are only a few possible types of ensembles to consider.  Specifically, we can classify the ensembles $\{\rho_i\}_{i=1}^n$ according to the ranks of their states, and so any ensemble of two-qubit states that can be perfectly distinguished belongs to one of the following classes: $\{1,1\}$, $\{1,1,1\}$, $\{1,1,1,1\}$, $\{1,2\}$, $\{1,3\}$, $\{1,1,2\}$ and $\{2,2\}$.  For instance, any ensemble of the type $\{1,1,2\}$ consists of three orthogonal states with respective ranks $1$, $1$, and $2$.  

In principle, the LOCC problem has been completely solved by the following lemma from Ref. \cite{Duan-2010a}.  In it, the notation $\Sch_\perp(\rho)$ denotes the minimal number of orthogonal product states whose linear span contains the support of $\rho$.
\begin{proposition}[\cite{Duan-2010a}]
\label{Prop:RunyaoMain}
A set of orthogonal two-qubit states $\{\rho_1,...\rho_n\}$ can be perfectly distinguished using LOCC if and only if 
\[
\sum_{i=1}^n \Sch_\perp(\rho_i)\leq 4.
\]
\end{proposition}
\noindent The goal of this section is to reformulate this lemma in a more transparent form for each class of two-qubit ensembles.  Additionally, we would like to compare this with the conditions needed for perfect discrimination by SEP.  This latter question has been solved, either explicitly or implicitly, in Ref. \cite{Duan-2007a} for all cases except the $\{1,2\}$ ensemble, i.e ensemble of one pure state and one rank two mixed state.  For completion, we will quickly review the LOCC and SEP discrimination conditions class-by-class and then end with a treatment of the $\{1,2\}$ case, since this is the previously missing piece in the perfect distinguishability picture.  Proofs are given for the cases which may not have been explicitly addressed before.

\subsubsection*{Case $\{1,1\}$ \cite{Walgate-2000a}} LOCC (hence SEP) distinguishability is always possible.

\subsubsection*{Cases $\{1,1,1,1\}$, $\{1,3\}$, $\{1,1,2\}$ \cite{Walgate-2002a}, \cite{Duan-2007a}}  LOCC and SEP distinguishability are both possible if and only if (iff) all the pure states are product states.  

\begin{IEEEproof}
The $\{1,1,1,1\}$ case is solved in Ref. \cite{Walgate-2002a} for LOCC and in Ref. \cite{Duan-2007a} for SEP.  

For the $\{1,3\}$ case, let $\ket{\psi}$ be the pure state in the ensemble.  The necessity of $\ket{\psi}$ being product state follows from the fact that the SEP POVM element\footnote{This implies the necessity for perfect LOCC discrimination.} detecting $\ket{\psi}$ must be the projector $\op{\psi}{\psi}$.  Conversely, when $\ket{\psi}=\ket{a}\ket{b}$, a perfect LOCC discrimination scheme consists of Alice and Bob measuring in the respective bases $\{\ket{a},\ket{\overline{a}}\}$ and $\{\ket{b},\ket{\overline{b}}\}$.  

For the $\{1,1,2\}$ case, we denote the pure states by $\ket{\psi_1}$ and $\ket{\psi_2}$, and as before, the POVM detecting each of these states must be $\op{\psi_1}{\psi_1}$ and $\op{\psi_2}{\psi_2}$ respectively.  Hence, for a SEP POVM, we must have that $\ket{\psi_1}$ and $\ket{\psi_2}$ are both product states.  The POVM element $\Pi_\rho$ detecting the mixed state $\rho$ must be a separable projector onto the support of $\rho$.  Thus, by Lemma \ref{Lem:MixedSep}, $\Pi_\rho=\op{x}{x}+\op{y}{y}$ for orthonormal product states $\ket{x}$ and $\ket{y}$.  But since $\rho$ is orthogonal to $\ket{\psi_1}$ and $\ket{\psi_2}$, the set $\{\ket{\psi_1},\ket{\psi_2},\ket{x},\ket{y}\}$ consisting of orthogonal product states can only be of the form $\{\ket{a}\ket{b},\ket{a}\ket{\overline{b}},\ket{\overline{a}}\ket{b},\ket{\overline{b}}\ket{\overline{b}}\}$.  Hence, a perfect LOCC discrimination scheme again consists of Alice and Bob measuring in the respective bases $\{\ket{a},\ket{\overline{a}}\}$ and $\{\ket{b},\ket{\overline{b}}\}$.
\end{IEEEproof}

\subsubsection*{Case $\{2,2\}$} SEP and LOCC distinguishability are both possible iff the projectors onto the supports of each of the mixed states are separable.  

\begin{IEEEproof}
The argument uses Lemma \ref{Lem:MixedSep} and follows analogously to the {case $\{1,1,2\}$}. 
\end{IEEEproof}

\subsubsection*{Case $\{1,1,1\}$ (\cite{Walgate-2002a}, \cite{Duan-2007a}} LOCC distinguishability is possible iff two of the three states are product states.  For SEP, let $\{\ket{\psi_i}\}_{i=1}^3$ denote the ensemble of three orthogonal states and $\ket{\Phi}$ the state orthogonal to all of them.  Then the $\ket{\psi_i}$ are perfectly
distinguishable by separable operations iff: (i)
$\psi_k\Phi^{-1}$ has two antiparallel eigenvalues for each entangled
state $\psi_i$, and (ii) $\sum_{i=1}^3 C(\psi_i) = C(\Phi)$.

\begin{IEEEproof}
The LOCC condition is given in Ref. \cite{Walgate-2002a} while the SEP criterion is proven in Ref. \cite{Duan-2007a}.  Examples of ensembles are presented in Ref. \cite{Duan-2007a} that satisfy the SEP distinguishability conditions but not the LOCC conditions.  Thus, SEP is strictly more powerful than LOCC for two-qubit state discrimination.  
\end{IEEEproof}

\subsubsection*{Case $\{1,2\}$}  Let $\ket{\psi}$ be orthogonal to a rank-two state $\rho$ with $\ket{\Phi}$ being orthogonal to both.  Then $\ket{\psi}$ and $\rho$ are perfectly distinguishable by SEP iff either $\ket{\psi}$ is a product state, or the two conditions hold: (i) the matrix $\psi\Phi^{-1}$ has two
antiparallel eigenvalues and (ii) $C(\psi)\leq C(\Phi)$. In particular, when $\Phi$ is a maximally entangled state, any such
$\ket{\psi}$ and $\rho$ are perfectly distinguishable by SEP.  For LOCC, the states are perfectly distinguishable iff either $\ket{\psi}$ is a product state, or condition (i) is satisfied and equality holds for condition (ii).  

\begin{IEEEproof}
We first consider LOCC.  Suppose that $\ket{\psi}$ is entangled and $\ket{\psi}$ and $\rho$ are perfectly distinguishable.  Then $\Sch_\perp(\psi)=2$ and Proposition~\ref{Prop:RunyaoMain} requires that $\Sch_\perp(\rho)=2$.  Hence, the assumption of Lemma \ref{Lem:MixedSep} is satisfied on the support of $\rho$, and thus conditions (i) and (ii) are satisfied, with the latter being an equality. Conversely, if the two conditions hold with $C(\psi)=C(\Phi)>0$, then Lemma \ref{Lem:Runyao} combined with Lemma \ref{Lem:MixedSep} implies that the support of $\rho$ has an orthogonal product basis.  Thus $\Sch_\perp(\psi)+\Sch_\perp(\rho)\leq 4$ and so LOCC discrimination is possible.  On the other hand, if $\ket{\psi}$ is a product state, a perfect discrimination protocol follows as in the $\{1,3\}$ case.   

Moving to SEP, suppose that the two states can be perfectly distinguished by
separable POVM $\{E,\mathbb{I}-E\}$. From Theorem 1 of \cite{Duan-2007a}, we know
that the operator $E$ must have the form:
$E=\proj{\psi}+\lambda\proj{\Phi}$, where $\lambda\leq 1$. Then
invoking Lemma~\ref{Lem:Runyao} directly gives the two conditions.  Conversely, whenever we have $\ket{\psi}$ and $\ket{\Phi}$
satisfying (i) and (ii), we can construct a separable operation
$E=\proj{\psi}+\lambda\proj{\Phi}$, where $\lambda=C(\psi)/C(\Phi)$,
according to Lemma~\ref{Lem:Runyao}. It is not difficult to see that
$\{E,\mathbb{I}-E\}$ can perfectly distinguish the state $\ket{\psi}$ and $\rho$.  One final
step is to verify that $\mathbb{I}-E$ is also separable.  As $E$ is a separable non-negative operator with rank two, it can be decomposed into
$E=x\proj{ab}+y\proj{cd}$, where $x,y$ are positive numbers and
$\ket{ab},\ket{cd}$ are two product states \cite{Sanpera-1998a}.  Next, observe that $E$ and
$E^{\Gamma}$ are equivalent up to a local unitary of form $I_A\otimes
U_B$. The unitary $U_B$ can be defined through $U_B\ket{b}=\ket{b^*}$
and $U_B\ket{d}=e^{i\theta}\ket{d^*}$, where $\ket{b^*}$ and
$\ket{d^*}$ are complex conjugate defined according to any fixed
orthonormal basis. That such a $U_B$ exists follows from the preservation of inner product: $|\langle b|d\rangle |=|\langle
b^*|d^*\rangle|$. So we see that the partial transpose of $\mathbb{I}-E$ is
locally equivalent to $\mathbb{I}-E$ itself and hence must be non-negative. And in the $2\otimes
2$ case, PPT of $\mathbb{I}-E$ implies its separability.

Finally, if $\ket{\Phi}$ is a maximally entangled state, then its
$2\times 2$ matrix representation $\Phi$ is equal to some $2\times 2$
unitary $U$: $\Phi=U$. Notice that
\[
\tr[\psi\Phi^{-1}] = \tr [U^\dagger \psi] = 2\langle \Phi|\psi\rangle =0,
\]
we can conclude that $\psi\Phi^{-1}$ would have two antiparallel
eigenvalues since $\ket{\psi}$ is entangled. Furthermore, since
$C(\Phi)=1$, the maximal possible value for concurrence, ii) holds.
This completes our proof.
\end{IEEEproof}

\section{Minimum Error Discrimination}

\label{Sect:MinError}

\subsection{General Linearly Independent Ensembles}

\label{Sect:Main}

We begin our discussion by recalling some general facts about minimum error discrimination.  For an ensemble $\{\rho_i,p_i\}_{i=1}^n$, a POVM $\{\Pi_i\}_{i=1}^n$ is optimal in minimum error discrimination if and only if $\Lambda-p_j\rho_j\geq 0$ for all $\rho_j$, in which the operator $\Lambda:=\sum_{i=1}^np_i\Pi_i\rho_i$ is hermitian \cite{Holevo-1973a, Yuen-1975a, Barnett-2009a}.  Since $\sum_{i=1}^n\Pi_i=\mathbb{I}$, we have
\begin{equation}
0=\tr[\Lambda-\Lambda]=\sum_{j=1}^n \tr[\Pi_j(\Lambda-p_j\rho_j)]=\sum_{j=1}^n \tr[(\Lambda-p_j\rho_j)\Pi_j].
\end{equation}
As $\Lambda-p_j\rho_j\geq 0$ and 
\begin{align*}
\tr[\Pi_j(\Lambda-p_j\rho_j)]&=\tr[(\Lambda-p_j\rho_j)\Pi_j]  \\ &=\tr[\Pi^{1/2}_j(\Lambda-p_j\rho_j)\Pi^{1/2}_j]\geq 0,
\end{align*}
we, in fact, must have  
\begin{equation}
\label{Eq:optimal-conds2}
\Pi_j(\Lambda-p_j\rho_j)=(\Lambda-p_j\rho_j)\Pi_j=0.
\end{equation}

Using these fundamental properties of the optimal POVM we next generalize a result given in Ref. \cite{Chitambar-2013b}, which itself is based on work by Mochon \cite{Mochon-2006a}.  This proposition offers the key tool used throughout most of this section.  

\begin{proposition}
\label{Prop1}
Let $\mathcal{E}=\{\rho_i,p_i\}_{i=1}^n$ be an ensemble of linearly independent states; i.e. for spectral decompositions $\rho_i=\sum_{j=1}^{r_i}\lambda_{ij}\op{\psi_{ij}}{\psi_{ij}}$, the $\ket{\psi_{ij}}$ are linearly independent.  Let $S$ be the subspace spanned by the $\ket{\psi_{ij}}$, and let $P_{opt}$ be the optimal minimum error probability in discrimination.   Then there exists a unique decomposition of $S=S_1\oplus S_2\oplus...\oplus S_n$ with $S_i$ having dimension $r_i$ such that a POVM can obtain $P_{opt}$ on $\mathcal{E}$ if and only if it can perfectly distinguish the normalized subspace projectors $\tfrac{1}{r_1}\Upsilon_{S_1}$, $\tfrac{1}{r_2}\Upsilon_{S_2}$,...,$\tfrac{1}{r_n}\Upsilon_{S_n}$.
\end{proposition}
\begin{proof}
Our argument proceeds like the one in Ref. \cite{Chitambar-2013b}.  As the $\ket{\psi_{ij}}$ are linearly independent, there exists a set of dual vectors $\ket{\widetilde{\psi}_{ij}}$ such that $\ip{\widetilde{\psi}_{ij}}{\psi_{k\ell}}=\delta_{ik}\delta_{j\ell}$.  Let $\{\Pi_i\}_{i=1}^n$ obtain $P_{opt}$ on $\mathcal{E}$.  Taking $\Upsilon_S$ to be the projector onto $S$ and defining $\widehat{\Pi}_j=\Upsilon_S\Pi_j\Upsilon_S$, it is obvious that the POVM $\{\widehat{\Pi}_j,\mathbb{I}-\Upsilon_S\}_{j=1}^n$ is also optimal for $\mathcal{E}$.  Thus the $\hat{\Pi}_j$ satisfy
\begin{equation}
\label{Eq:optimal-conds2-hat}
\widehat{\Pi}_j(\widehat{\Lambda}-p_j\rho_j)=(\widehat{\Lambda}-p_j\rho_j)\widehat{\Pi}_j=0,
\end{equation}
where $\widehat\Lambda=\Upsilon_S\Lambda \Upsilon_S$. 
We next observe that the vectors $\{\widehat{\Pi}_j\ket{\psi_{jk}}\}_{k=1}^{r_j}$ are linearly independent for each $j$.  For suppose that $\sum_{k=1}^{r_j}\alpha_k\widehat{\Pi}_j\ket{\psi_{jk}}=0$ for some nonzero $\alpha_k$.  Then we could contract both sides of $\widehat{\Lambda}-p_j\rho_j\geq 0$ with the vector $\ket{\varphi_j}=\sum_{k=1}^{r_j}\frac{\alpha_k}{p_j\lambda_{jk}}\ket{\widetilde{\psi}_{jk}}$ to obtain
  \[
  0\leq \bra{\varphi_j}  \left(\widehat{\Lambda}-p_j {\sum}_{k=1}^{r_j}\lambda_{jk}\op{\psi_{jk}}{\psi_{jk}}\right)\ket{\varphi_j}= -\sum_{k=1}^n\frac{|\alpha_k|^2}{p_j\lambda_{jk}}.
  \]
Hence, the space spanned by $\{\widehat{\Pi}_j\ket{\psi_{jk}}\}_{k=1}^{r_j}$ has dimension $r_j$.  Now, applying $\widehat{\Pi}_j(\widehat{\Lambda}-p_j\rho_j)=0$ in Eq. \eqref{Eq:optimal-conds2-hat} to $\ket{\widetilde{\psi}_{ik}}$ gives
\begin{equation}
\widehat{\Pi}_j(p_i\lambda_{ik}\widehat{\Pi}_i\ket{\psi_{ik}})=\delta_{ij}p_i\lambda_{ik}\widehat{\Pi}_i\ket{\psi_{ik}}.
\end{equation}
Thus, for every $1\leq k\leq r_i$, the element $\widehat{\Pi}_i\ket{\psi_{ik}}$ lies in the kernel of $\widehat{\Pi}_j$ for $i\not=j$, while $\widehat{\Pi}_i\ket{\psi_{ik}}$ is an eigenvector of $\widehat{\Pi}_j$ with eigenvalue +1 when $i=j$.  Thus, $\widehat{\Pi}_i$ must be the projector onto $S_i$, the $r_i$-dimensional subspace spanned by $\{\widehat{\Pi}_i\ket{\psi_{ik}}\}_{k=1}^{r_i}$.  Additionally, we have that $\widehat{\Pi}_i\widehat{\Pi}_j=\delta_{ij}\widehat{\Pi}_i$.  Clearly then for the original POVM $\{\Pi_i\}_{i=1}^n$ we have $\tr[\tfrac{1}{r_j}\Upsilon_{S_j}\Pi_i]=\delta_{ij}$, where $\Upsilon_{S_j}:=\widehat{\Pi}_j$ is the projector onto $S_j$.  Hence, the POVM can perfectly distinguish the normalized subspace projectors $\frac{1}{r_j}\Upsilon_{S_j}$.  Conversely, if a POVM $\{\Sigma_i\}_{i=1}^n$ perfectly distinguishes $\{\frac{1}{r_j}\Upsilon_{S_j}\}_{j=1}^n$, then we must have $\Upsilon_S\Sigma_i\Upsilon_S:=\widehat{\Sigma}_i=\Upsilon_{S_i}$, which means that the $\{\Sigma_i\}_{i=1}^n$ obtains $P_{opt}$ on $\mathcal{E}$.

To prove uniqueness, suppose that $\{\widehat{\Pi}_i\}_{i=1}^n$ and $\{\widehat{\Pi}'_i\}_{i=1}^n$ are two optimal POVMs on the subspace $S$.  Then any convex combination $\{\lambda\widehat{\Pi}_i+(1-\lambda)\hat{\Pi}_i'\}_{i=1}^n$ will also be optimal.  But as shown above, optimality requires that $\widehat{\Pi}_i$, $\widehat{\Pi}_i'$ and $\lambda\widehat{\Pi}_i+(1-\lambda)\widehat{\Pi}_i'$ are projectors onto some $r_i$-dimensional subspace of $S$.  This is possible only if $\widehat{\Pi}_i=\widehat{\Pi}_i'$.
\end{proof}

\subsubsection{Conditions of LOCC Optimality for Two-Qubit Pure States}

We now apply Proposition \ref{Prop1} to the LOCC discrimination of two-qubit pure ensembles $\{\ket{\psi_i}, p_i\}_{i=1}^n$.  As we only consider linearly independent ensembles, obviously $n\leq 4$.  In the pure state case, the subspace projectors $\Upsilon_{S_i}$ correspond to an orthonormal basis $\{\ket{\phi_i}\}_{i=1}^n$ for the space spanned by the $\ket{\psi_i}$.  Thus, any LOCC POVM $\Pi_i$ optimally distinguishes the $\ket{\psi_i}$ if and only if it can perfectly distinguish the $\ket{\phi_i}$.  However for two-qubit ensembles, the conditions for perfect discrimination among orthogonal states have already been proven by Walgate and Hardy \cite{Walgate-2002a}.  We thus obtain the following.

\begin{proposition}
\label{Prop2}
Consider an ensemble of linearly independent two-qubit states $\{\ket{\psi_i}, p_i\}_{i=1}^n$.  If $n=3$, then an LOCC protocol can optimally discriminate the ensemble (in the minimum error sense) if and only if the states $\{\ket{\phi_i}\}_{i=1}^n$ corresponding to the projectors $\Upsilon_{S_i}=\op{\phi_i}{\phi_i}$ described by Proposition \ref{Prop1} contain at least two product states.  If $n=4$, then all of the $\ket{\phi_i}$ must be product states.
\end{proposition}

Equivalently, for $n=3$ (resp. $n=4$) LOCC feasibility requires the decomposition $S=S_1\oplus..\oplus S_n$ of Proposition \ref{Prop1} to consist of at least two (resp. four) tensor product subspaces.  Applying this result is still rather difficult since there appears to be no easy method for determining whether or not the optimal POVM projectors $\op{\phi_i}{\phi_i}$ have product state form.  However, product POVMs belong to the more general class of separable POVMs, and for two-qubit systems, separability is captured by the PPT condition: a two-qubit positive operator $M$ is separable if and only if $M^\Gamma\geq 0$, where $\Gamma$ denotes the partial transpose operation \cite{Horodecki-1996a}.  Since we know that the optimal POVM is a unique rank one projective measurement, and the only PPT rank one projectors are product projectors, we thus obtain the following.
\begin{proposition}
\label{Prop:SDP}
Let $\{\ket{\psi_i},p_i\}_{i=1}^n$ an ensemble of linearly independent two-qubit states and $\Upsilon_S$ the projector on the subspace spanned by the $\ket{\psi_i}$.  Then the ensemble can be optimally distinguished by LOCC if and only if the following semi-definite program is feasible: Let $\Lambda = \sum_{j=1}^np_j\Pi_j\op{\psi_j}{\psi_j}$, $\forall i = 1, 2,\cdots, n$
\begin{align}
&\Lambda -p_i\proj{\psi_i} \geq 0,  \qquad \notag\\
\text{subject to}&  \notag\\
& \Pi_i \geq 0\ \text{and}\ \Lambda^\dagger=\Lambda\notag; 
\end{align}
\begin{align}
\text{In addition, if } n=3:&\notag\\
&(\Upsilon_S\Pi_\lambda\Upsilon_S)^\Gamma\geq 0\quad\text{for at least two $\lambda\in\{1,2,3\}$}.\notag\\
\text{if } n=4:&\notag\\
&(\Upsilon_S\Pi_\lambda\Upsilon_S)^\Gamma\geq 0\quad\text{for all $\lambda\in\{1,2,3,4\}$}.\notag
\end{align}
\end{proposition}

\subsubsection{LOCC is Not Optimal for Almost All Pure Three-State Ensembles}

\label{Sect:AlmostAll}

Proposition \ref{Prop:SDP} at least makes it a computable task to decide whether or not LOCC can achieve the optimal minimum error discrimination probability.  Nevertheless, we can use a randomized argument to show that for general instances, the SDP of Proposition \ref{Prop:SDP} is not feasible.  More precisely, if we randomly choose a two-qubit ensemble consisting of at least three states, almost surely there will be no LOCC protocol that optimally discriminates them.

We begin by fixing an algorithm for generating a random two-qubit ensemble of three pure states.  The first step consists of choosing a triple $(p_1,p_2,p_3)$ from the uniformly distributed probability simplex.  We thus assume that some distribution has been chosen and is fixed.  The protocol then consists of three independent samplings from $\mathcal{U}(2\otimes 2)$, the two-qubit unitary group, according to the uniform and unitarily invariant Haar distribution on $\mathcal{U}(2\otimes 2)$ \cite{Hall-2003a}.  Applying each of these unitaries to the state $\ket{00}$ generates an ensemble $\mathcal{E}$ of three two-qubit states.  Hence, the probability that a randomly chosen ensemble belongs to some (measurable) family of ensembles $\tau$ is $P(\tau)=\int\mathbbm{1}_\tau(\mathcal{E}) d\mathcal{E}$, where $\mathbbm{1}$ is the indicator function.  Given $\mathcal{E}$, we add one additional randomization step: another unitary $U$ is randomly chosen from $\mathcal{U}(2\otimes 2)$ and applied to $\mathcal{E}$.  In other words, we transform $\mathcal{E}\to U\mathcal{E}$, where $U\mathcal{E}$ is the ensemble obtained by applying $U$ to each element in $\mathcal{E}$.  In summary, letting $\mathcal{U}_\mathcal{E}:=\{U\mathcal{E}:U\in\mathcal{U}(2\otimes 2)\}$, our random ensemble is generated by first choosing $\mathcal{E}$ and then choosing a random $U\mathcal{E}\in\mathcal{U}_\mathcal{E}$.  Unitary invariance of the Haar measure ensures that $U\mathcal{E}$ has been uniformly selected among all possible ensembles of three states:
\begin{align}
P(\tau)=\int\mathbbm{1}_\tau(\mathcal{E}) d\mathcal{E}&=\int\int\mathbbm{1}_\tau(\mathcal{E}) d\mathcal{E}dU\notag\\
&=\int\int\mathbbm{1}_{U^\dagger\tau}(\mathcal{E}) d\mathcal{E}dU\notag\\
&=\int\int\mathbbm{1}_\tau(U\mathcal{E}) dUd\mathcal{E},
\end{align}
where the third equality follows from $P(\tau)=P(U^\dagger\tau)$.

The purpose of the extra randomization step is seen by the following lemma, which essentially reduces the problem of randomly choosing an ensemble to randomly choosing three orthogonal states.

\begin{lemma}
Let $\mathfrak{B}$ be the collection of all triples $(\ket{\phi_1},\ket{\phi_2},\ket{\phi_3})$ of pairwise orthonormal two-qubit states.  For each linearly independent ensemble $\mathcal{E}$ of three states, there is a bijection $\Delta:\mathcal{U}_\mathcal{E}\to\mathcal{B}$ such that $(\Delta(U\mathcal{E}))_{i=1}^3$ gives the optimal measurement basis for distinguishing the ensemble $U\mathcal{E}$.
\end{lemma}
\begin{proof}
For $\mathcal{E}=\{\ket{\psi_i},p_i\}_{i=1}^3$, Proposition \ref{Prop1} ensures the existence of states $(\ket{\phi_i})_{i=1}^3=:\Delta(\mathcal{E})\in\mathfrak{B}$ that form an optimal detection basis.  From the condition $\sum_{i=1}^3p_i\ket{\phi_i}\ip{\phi_i}{\psi_i}\bra{\psi_i}-p_k\op{\psi_k}{\psi_k}\geq 0$, we see that the optimal measurement basis for $U\mathcal{E}$ is given by the sequence $(U\ket{\phi_i})_{i=1}^3$, which is $\Delta(U\mathcal{E})\in\mathfrak{B}$.  Since $(U\ket{\phi_i})_{i=1}^3=(V\ket{\phi_i})_{i=1}^3$ iff $V^\dagger U$ is the identity, the mapping $\Delta$ is injective.  Conversely, any element $(\ket{\phi'_i})_{i=1}^3$ in $\mathfrak{B}$ can be related to $\Delta(\mathcal{E})=(\ket{\phi_i})_{i=1}^3$ by a unitary $V$, and so the measurement basis given by $(\ket{\phi_i'})_{i=1}^3$ is optimal for the ensemble $V\mathcal{E}\in\mathcal{U}_\mathcal{E}$.  Hence, $\Delta$ is surjective.
\end{proof}

Putting it all together, for a fixed distribution $(p_1,p_2,p_3)$, let $L$ be the set of all two-qubit linearly independent ensembles that can be perfectly distinguished by LOCC.  Since the subset of all linearly dependent ensembles has measure $0$, it suffices to only consider the probability of randomly choosing an ensemble belonging to $L$.  We have $P(L)=\int\int\mathbbm{1}_{L}(U\mathcal{E}) dUd\mathcal{E}$, and by Proposition \ref{Prop2} and the previous Lemma, $\int\mathbbm{1}_{L}(U\mathcal{E}) dU$ is precisely the probability of randomly choosing an orthonormal triple $(\ket{\phi_1},\ket{\phi_2},\ket{\phi_3})$ such that at least two of them are product states.  This occurs with probability zero, and thus we prove the following.
\begin{theorem}
Three randomly chosen two-qubit pure states almost surely cannot be discriminated optimally by LOCC.
\end{theorem}

\subsection{Discrimination of Ensembles that Lack Entanglement}

It is especially interesting to consider ensembles that do not possess any entanglement.  When each of the states are product states, one might naively suppose that the ensemble also lacks any sort of nonlocality, and thus LOCC should be able to achieve optimal discrimination.  However, this intuition turns out to be incorrect as there exists product state ensembles that cannot be optimally distinguished by LOCC \cite{Bennett-1999a}.  Below, we use Proposition \ref{Prop1} to prove new instances of nonlocality without entanglement in two-qubit systems.  

\subsubsection{Three States}

\label{Sect:3states}

First we turn to ensembles of three product states.  The following provides a necessary condition for optimal discrimination by LOCC.

\begin{theorem}
\label{Thm:3prodstates}
Suppose that $\{\ket{\psi_\lambda}:=\ket{\alpha_\lambda}\ket{\beta_\lambda},p_\lambda\}_{\lambda=1}^3$ ($p_\lambda>0$) is some linearly independent two-qubit product state ensemble that spans $\{\ket{\Phi}\}^\perp$.  Let $\lambda_{min}(\Phi)$ denote the smallest squared Schmidt coefficient of $\ket{\Phi}$.  If 
\[p_i^2\lambda_{min}^2(\Phi)>p_j^2|\ip{\psi_i}{\psi_j}|^2+p_k^2|\ip{\psi_i}{\psi_k}|^2 \] for every choice of $i,j,k$ such that $\{i,j,k\}=\{1,2,3\}$, then the ensemble cannot be distinguished optimally (in the minimum error sense) with LOCC.
\end{theorem}

\begin{proof}
Following Proposition 2, if there exists an LOCC protocol that obtains optimality, then we have that $\ket{\phi_j}$ and $\ket{\phi_k}$ are product states for some $j\not= k\in\{1,2,3\}$.  It is then easy to show \cite{Duan-2007a} that $C(\phi_i)=C(\Phi)$, where $C(\phi_i)$ is the concurrence of $\ket{\phi_i}$ \cite{Wootters-1998a}.  Since $\Phi$ and $\phi_i$ have the same concurrence, they also have the same Schmidt coefficients.  Optimality requires
\[\sum_{\lambda=j,k} p_\lambda\ket{\phi_\lambda}\ip{\phi_\lambda}{\psi_\lambda}\bra{\psi_\lambda}+p_i\ket{\phi_i}\ip{\phi_i}{\psi_i}\bra{\psi_i}\geq p_i\op{\psi_i}{\psi_i}.\]
Contracting both sides by $\ket{\psi_i}$ gives 
\begin{equation}\label{condi01}
\sum_{\lambda=j,k} p_\lambda\ip{\psi_i}{\phi_\lambda}\ip{\phi_\lambda}{\psi_\lambda}\ip{\psi_\lambda}{\psi_i}\geq p_i(1-|\ip{\psi_i}{\phi_i}|^2).
\end{equation}
Applying the Cauchy-Schwarz inequality on the LHS of (\ref{condi01}) gives
\begin{align}\label{condi02}
&\left|\sum_{\lambda=j,k} p_\lambda\ip{\psi_i}{\phi_\lambda}\ip{\phi_\lambda}{\psi_\lambda}\ip{\psi_\lambda}{\psi_i}\right| \notag\\
&\leq\sqrt{\sum_{\lambda=j,k}|\ip{\psi_i}{\phi_\lambda}|^2}\sqrt{\sum_{\lambda=j,k}p_\lambda^2|\ip{\psi_i}{\psi_\lambda}|^2}\notag\\
&=\sqrt{1-|\ip{\psi_i}{\phi_i}|^2} \sqrt{\sum_{\lambda=j,k}p_\lambda^2|\ip{\psi_i}{\psi_\lambda}|^2}\notag \\
&\leq \sqrt{\sum_{\lambda=j,k}p_\lambda^2|\ip{\psi_i}{\psi_\lambda}|^2}.
\end{align}
Hence, Eqs. (\ref{condi01})-(\ref{condi02}) give
\begin{align}
\sqrt{\sum_{\lambda=j,k}p_\lambda^2|\ip{\psi_i}{\psi_\lambda}|^2}&\geq p_i \left(1-|\ip{\psi_i}{\phi_i}|^2 \right)\notag\\ 
&\geq p_i \lambda_{min}(\Phi)
\end{align}
where the second inequality follows from the fact that $\ket{\psi_i}$ is a product state and therefore its overlap with $\ket{\phi_i}$ can be no greater than the largest Schmidt coefficient of $\ket{\phi_i}$.  This gives the desired necessary condition for LOCC discrimination.
\end{proof}

Theorem \ref{Thm:3prodstates} is very useful for constructing ensembles that demonstrate ``non-locality without entanglement''.  Despite consisting of product states, ensembles satisfying the condition of Theorem \ref{Thm:3prodstates} possess some non-local aspect since LOCC is insufficient for optimal discrimination.  Furthermore, we can obtain examples in which separable operations attain optimal discrimination but LOCC cannot.  For this, we rely on the fact that, as mentioned in Section \ref{Sect:2x2Perfect} three states can be perfectly distinguished by separable operations iff their concurrence sums to the concurrence of their common orthogonal complement state.  Hence, separable operations become strictly more powerful than LOCC for distinguishing a set of product states $\ket{\psi_i}$ that satisfy Theorem \ref{Thm:3prodstates} while their corresponding detection states $\ket{\phi_i}$ satisfy $\sum_{i=1}^3C(\phi_i)=C(\Phi)$.

An important example of such an ensemble is the so-called ``double trine'' ensemble \cite{Peres-1991a}, which is given by a uniform distribution of the states $\ket{\psi_i}=\ket{s_i}\otimes\ket{s_i}$ for $i=0,1,2$ where
\begin{align}\label{eq-doubletrine}
\ket{s_0}&=\ket{0}, \notag \\
\ket{s_1}&=-\frac{1}{2}\ket{0}-\frac{\sqrt{3}}{2}\ket{1}, \notag \\
\ket{s_2}&=-\frac{1}{2}\ket{0}+\frac{\sqrt{3}}{2}\ket{1}.
\end{align}
The inability for LOCC to optimally discriminte the double trine states follows from Theorem \ref{Thm:3prodstates} and the fact that $|\ip{\psi_i}{\psi_j}|^2=1/16$, while $\lambda_{min}^2(\Phi)=1/4$.  Thus, $1/4>1/8$.  On the other hand, it can be easily computed \cite{Peres-1991a, Ban-1997a, Chitambar-2013b} that the detection states $\ket{\phi_i}$ each have a concurrence of $1/3$.  Since the maximally entangled singlet state lies orthogonal to each of the $\ket{\psi_i}$, we indeed have $\sum_{i=1}^3C(\phi_i)=C(\Phi)=1$.  Hence, a separable POVM can optimally discriminate the double trine ensemble.  In Ref. \cite{Chitambar-2013b}, an even stronger result was proven that a finite gap exists between the best achievable LOCC success probability and the optimal probability achievable by \textit{any} physically implementable LOCC operation.  In Section \ref{Sect:Ncopy}, we generalize the double trine ensemble to its $N$-copy form.

\subsubsection{Discrimination of One Pure Product State and One Rank Two Separable State}

\label{Sect:PureProd}

In this section, we apply Proposition \ref{Prop1} to a mixed state discrimination problem.  Consider the following orthonormal basis $\{\ket{\Psi^-},\ket{\Psi^+},\ket{s^+(\theta)},\ket{s^-(\theta)}\}$ where 
\begin{align}
\ket{\Psi^\pm}&=\sqrt{1/2}(\ket{01}\pm\ket{10}),\notag\\
\ket{s^+(\theta)}&=\cos\theta\ket{00}+\sin\theta\ket{11},\notag\\
\ket{s^-(\theta)}&=-\sin\theta\ket{00}+\cos\theta\ket{11}
\end{align} for $0\leq\theta\leq \pi/2$.  The task is to optimally distinguish the following two states:
\begin{align}
\op{\psi}{\psi}&=\op{00}{00},\notag\\
\rho(\theta)&=\frac{1}{1+\sin 2\theta}(\op{s^+(\theta)}{s^+(\theta)}+\sin 2\theta\op{\Psi^+}{\Psi^+}).
\end{align}
The concurrence of $\ket{s^{\pm}(\theta)}$ is $1+\sin 2\theta$, and so by Lemma \ref{Lem:Runyao}, $\rho(\theta)$ is separable for any choice of $\theta$.  We note that
\begin{equation}
\label{Eq:Koashi}
\rho(\pi/4)=\frac{1}{2}(\op{++}{++}+\op{--}{--}),
\end{equation}
and the states reduce to the pair studied in Ref. \cite{Koashi-2007a}.
\begin{lemma}
For $0<\theta<\pi/2$, the states $\ket{00}$ and $\rho(\theta)$ given with arbitrary \textit{a priori} probability can be distinguished optimally by SEP but not by LOCC.
\end{lemma}
\begin{proof}
We first use Proposition \ref{Prop1}, to translate the problem into one of perfect discrimination between orthogonal states.  We can easily decompose the hermitian operator $p_1\rho(\theta)-p_2\op{00}{00}$ into its positive and negative eigenspaces:
\[
p_1\rho(\theta)-p_2\op{00}{00} =a\op{\Psi^+}{\Psi^+}+b\op{s^+(\bar{\theta})}{s^+(\bar{\theta})}-c\op{s^-(\bar{\theta})}{s^-(\bar{\theta})}
\]
for positive numbers $a,b,c$ and with $\bar{\theta}\not=0,\pi/2$, since $0<\theta<\pi/2$ (an explict formula for $\bar{\theta}$ can be computed but is irrelevant for our purpose).  Optimal discrimination of $\ket{00}$ and $\rho(\theta)$ thus amounts to perfect discrimination of the state $\ket{s^-(\bar{\theta})}$ and the normalized projector $\tfrac{1}{2}\Upsilon$, where $\Upsilon=\op{\Psi^+}{\Psi^+}+\op{s^+(\bar{\theta})}{s^+(\bar{\theta})}$.  We have that $C(s^-(\bar{\theta}))<1$ which, by case $\{1,2\}$ of Section \ref{Sect:2x2Perfect}, implies that SEP achieves perfect discrimination but LOCC does not.
\end{proof}

To our knowledge, this is the first time a pair of states have been shown to be optimally distinguishable, in the minimum error sense, by SEP but not LOCC.  Ref.~\cite{Koashi-2007a} shows a distinguishability gap for $\ket{00}$ and the mixed state given by Eq. \eqref{Eq:Koashi}, but with respect to a special type of unambiguous discrimination measure. 

\subsubsection{$N$-Copy Trine Ensemble}
\label{Sect:Ncopy}
Consider the equiprobable ensemble of $N$-qubit states\footnote{When $N=2$, it corresponds to the double trine ensemble given in (\ref{eq-doubletrine}).} $\{\ket{\psi_i}\}_{i=0}^2$ with $\ket{\psi_i}=(U^i\ket{0})^{\otimes N}$, where
\[
U=-e^{2\pi i/N}\left(\begin{smallmatrix}1/2&-\sqrt{3/4}\\\sqrt{3/4}&1/2\end{smallmatrix}\right).
\] 
Note that 
\[
U^2=e^{4\pi i/N}\left(\begin{smallmatrix}-1/2&-\sqrt{3/4}\\\sqrt{3/4}&-1/2\end{smallmatrix}\right),
\]
and so the global phase of $U$ is chosen such that $(U^{\otimes N})^3=\mathbb{I}$.  Also, we have $\ip{\psi_i}{\psi_j}=(-1/2)^N$ for $i\not=j$.
 
We want to show that these states cannot be optimally distinguished using $N$-party LOCC.  To accomplish this, we will map these three states to an ensemble of qu\textit{trit} states known as the ``lifted'' trine ensemble \cite{Shor-2002a}:
\begin{align*}
\ket{L_0(\alpha)}&=\sqrt{1-\alpha}\ket{0}+\sqrt{\alpha}\ket{2}\\
\ket{L_1(\alpha)}&=\sqrt{1-\alpha}(-\tfrac{1}{2}\ket{0}-\tfrac{\sqrt{3}}{2}\ket{1})+\sqrt{\alpha}\ket{2}\\
\ket{L_2(\alpha)}&=\sqrt{1-\alpha}(-\tfrac{1}{2}\ket{0}+\tfrac{\sqrt{3}}{2}\ket{1})+\sqrt{\alpha}\ket{2}.
\end{align*} 
The value of $\alpha$ is known as the ``lifting angle'' of the ensemble.  
 
\begin{lemma}\label{lem-iso}
For any fixed $N$, there is an isomorphic mapping $\varphi:\ket{\psi_i}=(U^i\ket{0})^{\otimes N}\to \ket{L_i(\alpha)}$ with 
\[
\alpha=\frac{1}{3}\left(1-(-\frac{1}{2})^{N-1}\right),
\] 
for $i=0,1,2$.
\end{lemma} 
\begin{IEEEproof} 
To construct this equivalence, first let 
\begin{equation}
\label{Eq:z}
\ket{z}=(3\kappa)^{-1/2}\sum_{i=0}^2\ket{\psi_i}
\end{equation}
where $\kappa$ is a normalization constant that will be given later.
This defines our ``axis'' of rotation similar to $\ket{2}$ in the lifted trine.  Using the symmetry we have that 
\begin{align*}
\ip{\psi_i}{z}&=(3\kappa)^{-1/2}\sum_{j=0}^2\bra{\psi_0}(U^{j})^{\otimes N}\ket{\psi_0}\\
&=(3\kappa)^{-1/2}(1+2(-1/2)^N),
\end{align*}
where $\kappa=1-(-1/2)^{N-1}$ is defined so that 
$\ip{z}{z}=(3\kappa)^{-1}[3(1+2(-1/2)^N)]=1.$   
We also have $\ip{\psi_i}{z}=\sqrt{\kappa/3}$.
We next define 
\[
\ket{\omega_i}:=\frac{\ket{\psi_i}-\ket{z}\ip{z}{\psi_i}}{\sqrt{1-\kappa/3}}
\]
 so that $\ip{\psi_i}{\omega_i}=\sqrt{1-\kappa/3}$.  Therefore,
\[\ket{\psi_i}=\sqrt{1-\kappa/3}\ket{\omega_i}+\sqrt{\kappa/3}\ket{z}.\]
We have 
\begin{align*}
\ip{\omega_i}{\omega_j}&=\frac{\ip{\psi_i}{\psi_j}-\ip{\psi_i}{z}\ip{z}{\psi_j}}{1-\kappa/3}\\
&=\frac{(-1/2)^N-\kappa/3}{1-\kappa/3}=-1/2.
\end{align*}
The $\{\ket{\omega_i}\}$ satisfy the important property that $\sum_{i=0}^2\ket{\omega_i}=0$, and thus they constitute a linearly dependent set.  We can then establish an isomorphic mapping $\varphi$ by $\ket{\omega_i}\leftrightarrow U^i\ket{0}$ and $\ket{z}\leftrightarrow \ket{2}$.  As a result, the $N$-copy trine states can be identified with the lifted trine states $\ket{\psi_i}\leftrightarrow\ket{L_i(\alpha)}$ with lifting angle $\alpha=\kappa/3=1/3(1-(-1/2)^{N-1})$.  
\end{IEEEproof}

\begin{theorem}
For any finite $N$, $\{\ket{\psi_i}\}_{i=0}^2$ cannot be optimally distinguished by LOCC.
\end{theorem}
\begin{IEEEproof}
Via the isomorphism $\varphi$ given in Lemma~\ref{lem-iso}, the optimal POVM for distinguishing the $\ket{\psi_i}$ can be found by solving the problem for the $\ket{L_i(\alpha)}$.  Since the $\ket{L_i(\alpha)}$ represent a $U$-covariant ensemble, the optimal measurement is given by the so-called ``pretty good measurement'' (PGM)\footnote{Recall that the ``Pretty Good Measurement'' for an ensemble $\{\ket{\phi_i},p_i\}_{i=1}^k$ is the POVM with elements \unexpanded{$p_i\rho^{-1/2}\op{\phi_i}{\phi_i}\rho^{-1/2}$, where $\rho=\sum_{i=1}^kp_i\op{\phi_i}{\phi_i}$} \cite{Hausladen-1994a}.}.  The optimal POVM is given by the orthogonal projectors $\{\op{f_i}{f_i}\}_{i=0}^2$ where
\begin{equation*}
\ket{f_i}:=\sqrt{\frac{2}{3}}U^i\ket{0}+\sqrt{\frac{1}{3}}\ket{2}\leftrightarrow \ket{F_i}:=\sqrt{\frac{2}{3}}\ket{\omega_i}+\sqrt{\frac{1}{3}}\ket{z}.
\end{equation*}
The success probability is given by 
\[
\frac{1}{3}\sum_{i=0}^2|\ip{F_i}{\psi_i}|^2=\left(\sqrt{\frac{2}{3}}\sqrt{1-\kappa/3}+\sqrt{\frac{1}{3}}\sqrt{\kappa/3}\right)^2.
\]  
Note that this goes to 1 as $N\to\infty$, which reflects the fact that the $\ket{\psi_i}$ become orthogonal in the asymptotic limit.

We now want to prove that the $\ket{F_i}$ cannot be perfectly distinguished by $N$-party LOCC.  By Proposition \ref{Prop1}, this sufficies for proving that LOCC is unable to optimally distinguish $N$-copies of the trine states.  Following the work of Ref. \cite{Walgate-2002a}, the POVM states can be perfectly distiniguished only if there exists an orthornormal basis $\{\ket{b_0},\ket{b_1}\}$ for some party such that, {for $i=0,1,2$} 
\[
\ket{F_i}=\sqrt{2/3}\ket{\omega_i}+\sqrt{1/3}\ket{z}=\ket{b_0}\ket{\eta^i_0}+\ket{b_1}\ket{\eta^i_1}
\]
and $\ip{\eta^i_0}{\eta^j_0}=\ip{\eta^i_1}{\eta^j_1}=0$ for $i\not=j$.  Summing over the $\ket{F_i}$ and using the fact that $\sum_{i=0}^2\ket{\omega_i}=0$ gives $\ket{z}=\sqrt{1/3}(\ket{b_0}\ket{v_0}+\ket{b_1}\ket{v_1})$ where $\ket{v_0}=\sum_{i=0}^2\ket{\eta^i_0}$ and $\ket{v_1}=\sum_{i=0}^2\ket{\eta^i_1}$.  It is clear that $\ket{v_0}$ and $\ket{v_1}$ are linearly independent or else by the overall symmetry, $\ket{z}$ would be an $N$-partite product state, which it's not (this can be seen by simply contracting $\ket{z}$ in Eq. \eqref{Eq:z} with $\ket{0}$ for all but two parties and observing that the remaining bipartite state is entangled).  Furthermore, from the definition of $\ket{z}$, we have that both $\ip{b_0}{z}$ and $\ip{b_1}{z}$ are invariant under the action of $U^{\otimes N-1}$.  Hence so are $\ket{v_0}$ and $\ket{v_1}$.  But then the invariance of $\ket{z}$ under $U^{\otimes N}$ and the linear independence of the $\ket{v_0}$ and $\ket{v_1}$ imply that $\ket{b_0}=U\ket{b_0}$, which is impossible.  Therefore, the $\ket{F_i}$ cannot be decomposed as required for LOCC distinguishability.  
\end{IEEEproof}

What's particularly interesting about the $N$-copy trine ensemble is that LOCC \textit{can} distinguish the states optimally as $N\to\infty$.  To see this, consider the single party POVM consisting of three projectors $\{\op{\phi_i}{\phi_i}\}_{i=0}^2$ where $\ket{\phi_i}=U^i\ket{1}$.  Note that for each outcome, one of the the trine states is eliminated.  Hence, consider the LOCC protocol where each party performs this measurement and then globally communicate their results.  The probability that they are unable to eliminate two of the states (i.e. they all eliminate the same state) is $(1/3)^N\to 0$.

This result is quite interesting when one considers the $N$-copy problem for two pure state ensembles.  It has been proven that $N$-party LOCC can always obtain the optimal success probability \cite{Acin-2005a}.  

\subsubsection{$N$-Ensemble Copies}

\label{Sect:N Ensemble}

We finally consider ensembles of product states that can be formed by taking $N$ copies of one particular ensemble.  In other words, the ensemble has the decomposition $\{\rho_i, p_i\}_{i=1}^{k^N}=\big(\{\sigma_i,q_i\}_{i=1}^k\big)^{\otimes N}$ (compare this to the previous section).  The following is a very simple observation.  Let $\{\Pi_i\}^k_{i=1}$ be the optimal POVM for the underlying ensemble $\{\sigma_i,q_i\}_{i=1}^k$.  Then it is clear that $\Lambda\geq q_i\sigma_i$ for all $i$ implies that $\Lambda^{\otimes 2}\geq q_iq_j\sigma_i\otimes \sigma_j$ for all $i,j$.  By induction we then see that  $\big(\{\Pi_i\}_{i=1}^k\big)^{\otimes N}$ is the optimal POVM for the ensemble $\{\rho_i, p_i\}_{i=1}^{k^N}$.  In other words, in the $N$-party LOCC setting, the ensemble $\{\rho_i, p_i\}_{i=1}^{k^N}$ can be optimally distinguished by each party simply perfoming the same local POVM $\{\Pi_i\}_{i=1}^k$.  In fact, the same argument shows that we don't need the overall ensemble to be $N$ copies of the same underlying ensemble; rather it needs to only be the tensor product of $N$ ensembles.

\section{Unambiguous Discrimination}
\label{Sect:Unambig}

We now consider the task of unambiguous discrimination for two-qubit pure states.  Recall that unambiguous discrimination of an ensemble $\{\ket{\psi_i},p_i\}_{i=1}^n$ consists of a POVM $\{\Pi_i\}_{i=0}^{n}$ with $n+1$ outcomes.  Each state is assigned one of the $n$ outcomes, and the remaining outcome $\Pi_0$ corresponds to an inconclusive or ambiguous conclusion.  The constraint is that for each of the unambiguous outcomes, there is no decision error.  In other words, $\bra{\psi_i}\Pi_j\ket{\psi_i}=0$ for $i\not=j>0$.  The task is to choose a POVM that minimizes the probability of an inconclusive outcome.  Explicitly, the problem is
\begin{align}
\label{Eq:Cons}
\min_\Pi\quad& q=\sum_{i=1}^np_i\bra{\psi_i}\Pi_0\ket{\psi_i}\notag\\
s.t.\quad&\bra{\psi_i}\Pi_j\ket{\psi_i}=0\quad i\not=j>0.
\end{align}
Not all ensembles will allow for a feasible solution, and unambiguous discrimination is possible if and only if the states are linearly independent \cite{Chefles-1998b}.  For such ensembles, the elements $\Pi_i$ are easy to characterize.  Let $S$ be the subspace spanned by the $\ket{\psi_i}$, and take $\ket{\varphi_i}$ to be an orthonormal basis of $S$.  With respect to this basis, let $R$ be the $dim(S)\times dim(S)$ matrix whose columns are the $\ket{\psi_i}$.  Linear independence of the $\ket{\psi_i}$ ensures invertibility of $R$.  After normalization, the rows of $R^{-1}$ are the states $\ket{\widetilde{\psi}_i}$ with the property that $\ip{\widetilde{\psi}_i}{{\psi}_j}=0$ for $i\not=j$.  Furthermore, up to an overall factor, these are the unique vectors in $S$ with this property\footnote{Note, here we demand that the dual states $\ket{\widetilde{\psi}_i}$ are normalized in contrast to the vectors used in the proof of Proposition \ref{Prop1}.}.  As a result, the POVM element $\Pi_i$ for unambiguous discrimination must have support on $\widetilde{S_i}$, which is the space spanned by $\ket{\widetilde{\psi}_i}$ and elements lying in $S^\perp$.  

For unambiguous discrimination by LOCC, there is an added constraint to \eqref{Eq:Cons} that each of the POVM must be realized by an LOCC protocol.  In general, this added constraint will often make the problem infeasible.  For general pure state ensembles, Chefles has shown that unambiguous discrimination is possible by LOCC iff for each $\ket{\psi_i}$, there exists a product state in $\mathcal{S}_i^\perp$  \cite{Chefles-2004a}.  Subsequent work on LOCC unambiguous discrimination was conducted in References \cite{Duan-2007b, Bandyopadhyay-2009a}.  

When the ensemble under consideration consists of four two-qubit pure states, unambiguous discrimination by LOCC becomes possible iff each of the $\ket{\widetilde{\psi_i}}$ are product states \cite{Duan-2007b}.  Of course, four orthogonal product state can be perfectly distinguished by LOCC.  But a non-trivial example consists of the ensemble
\begin{align}
\ket{\psi_1}&=\ket{00},&\ket{\psi_2}&=\ket{0+},&\ket{\psi_3}&=\ket{+0},&\ket{\psi_4}&=\ket{++},
\end{align}
where $\ket{\pm}=\sqrt{1/2}(\ket{0}\pm\ket{1})$.  The states necessary for LOCC unambiguous discrimination are
\begin{align}
\ket{\widetilde{\psi}_1}&=\ket{--},&\ket{\widetilde{\psi}_2}&=\ket{-1},&\ket{\widetilde{\psi}_3}&=\ket{1-},&\ket{\widetilde{\psi}_4}&=\ket{11}.
\end{align}
On the other hand, the linearly independent states 
\begin{align}
\label{Eq:4unambig}
\ket{\psi_1}&=\ket{00},&\ket{\psi_2}&=\ket{0+},&\ket{\psi_3}&=\ket{+0},&\ket{\psi_4}&=\ket{11}
\end{align}
cannot be unambiguously discriminated by LOCC, while they can be by a global POVM.  To see LOCC impossibility directly, note that the state lying orthogonal to $\ket{\psi_2}$, $\ket{\psi_3}$ and $\ket{\psi_4}$ is $\sqrt{1/3}(\ket{00}+\ket{01}+\ket{10})$, which is entangled.  Thus, we can conclude that ensemble \eqref{Eq:4unambig} demonstrates a type of nonlocality without entanglement, at least with respect to unambiguous discrimination \cite{Duan-2007b}.

In contrast to four-state ensembles, the ability to unambiguously discriminate three states by global operations implies feasibility by LOCC \cite{Bandyopadhyay-2009a}.  This follows from the fact that any two-dimensional subspace in $\mathbb{C}^2\otimes\mathbb{C}^2$ contains at least one product state \cite{Sanpera-1998a}.  Nevertheless, there is still the question of whether the maximum global unambiguous probability $1-q$ can be achieved by LOCC.

\subsection{Symmetric Ensembles}
 
For symmetric product state ensembles, we can obtain an upper bound on the conclusive probability.  By symmetric states, we mean those that are invariant under the SWAP operation $\mathbb{F}$, which acts on any product state $\ket{\alpha\beta}$ by $\mathbb{F}\ket{\alpha\beta}=\ket{\beta\alpha}$.

\begin{theorem}
\label{Thm:unambigSymm}
Let $\{\ket{\psi_i},p_i\}_{i=1...3}$ be an ensemble of two-qubit linearly independent symmetric pure states with $\ket{\widetilde{\psi}_i}$ being dual states satisfying $\ip{\widetilde{\psi}_i}{\psi_j}=0$ for $i\not=j$.  If $C(\widetilde{\psi}_i)\geq |\ip{\widetilde{\psi}_i}{\psi_i}|^2$ for all $i$, then LOCC cannot obtain an unambiguous probability greater than $p_{max}:=\max\{p_1,p_2,p_3\}$.
\end{theorem}
\begin{proof}
We will prove that this theorem holds for the more general class of separable operations.  Let $\ket{\Psi^-}=\frac{1}{\sqrt{2}}(\ket{01}-\ket{10})$ be the anti-symmetric state lying orthogonal to the ensemble states.  Then, the conclusive POVM elements $\Pi_i$ must take the form
\[
\Pi_i=a_i\op{\widetilde{\psi_i}}{\widetilde{\psi_i}}+b_i\op{\Psi^-}{\widetilde{\psi_i}}+b_i^*\op{\widetilde{\psi_i}}{\Psi^-}+c_i\op{\Psi^-}{\Psi^-}.
\]
The total success probability of this POVM is given by $\sum_{i=1}^3a_ip_i|\ip{\widetilde{\psi_i}}{\psi_i}|^2$.  Note that the $\ket{\widetilde{\psi_i}}$ lie in the symmetric subspace.   Our first task is to show that we can take $b_i=0$ without loss of generality.  For the $\Pi_i$ to be an separable POVM, we need that the $\Pi_i$ can be expressed as a positive sum of product states.  Then if $\Pi_i$ is separable, so is $\mathbb{F}\Pi_i\mathbb{F}$ as well as the group projection $\tau(\Pi_i):=(\Pi_i+\mathbb{F}\Pi_i\mathbb{F})/2$.  Furthermore, we have $\sum_{i=1}^3\tau(\Pi_i)\leq\mathbb{I}$ and also $|\bra{\widetilde{\psi}_i}\mathbb{F}\ket{\psi_i}|^2=|\ip{\widetilde{\psi_i}}{\psi_i}|^2$.  Therefore, we can replace the separable POVM $\Pi_i$ with the separable POVM $\tau(\Pi_i)=a_i \op{\widetilde{\psi_i}}{\widetilde{\psi_i}}+c_i\op{\Psi^-}{\Psi^-}$, and the overall conclusive probability remains unchanged.  

Next, we compute the required values of $a_i$ and $c_i$ for each $\Pi_i$ to be separable.  It is not difficult to verify that, up to an overall constant, the only two product states lying in the span of $\ket{\widetilde{\psi}_i}$ and $\ket{\Psi^-}$ are $\ket{\widetilde{\psi}_i}\pm C(\widetilde{\psi}_i)\ket{\Psi^-}$, where $C(\widetilde{\psi}_i)$ is the concurrence of $\ket{\widetilde{\psi_i}}$.  Thus, the separable $\Pi_i$ must take the form $\Pi_i=\alpha_i( \op{\widetilde{\psi}_i}{\widetilde{\psi}_i}+C(\widetilde{\psi}_i)\op{\Psi^-}{\Psi^-})$.  For $\sum_{i=1}^3\Pi_i\leq\mathbb{I}$, we need that $\sum_{i=1}^3\alpha_iC(\widetilde{\psi_i})\leq 1$.  On the other hand, the overall conclusive probability is $\sum_{i=1}^3p_i\alpha_i|\ip{\widetilde{\psi}_i}{\psi_i}|^2$.  Hence, if $C(\widetilde{\psi}_i)\geq |\ip{\widetilde{\psi}_i}{\psi_i}|^2$ for all $i$, then the total conclusive probability will be no greater than $p_{max}$.
\end{proof}
As a simple example of this is the equiprobable ensemble of the three symmetric Bell states $\sqrt{1/2}(\ket{00}\pm\ket{11})$ and $\sqrt{1/2}(\ket{01}+\ket{10})$ \cite{Ghosh-2001a}.  Here $\ket{\widetilde{\psi_i}}=\ket{\psi_i}$, and the conditions of Theorem \ref{Thm:unambigSymm} are met.  Hence, the LOCC conclusive probability cannot exceed $1/3$ while the global conclusive probability is $1$.

\subsection{The Double Trine Ensemble}

The converse to Theorem \ref{Thm:unambigSymm} does not hold in general.  As an interesting example, we consider the double trine ensemble (\ref{eq-doubletrine}), and show that LOCC and SEP obtain the same maximum conclusive probability, which turns out to be less than the optimal probability feasible by global operations.  Thus, the double trine ensemble demonstrates a very curious distinguishability property: For minimum-error discrimination, LOCC $<$ SEP $=$ GLOBAL; For optimal unambiguous discrimination, LOCC $=$ SEP $<$ GLOBAL.

\subsubsection*{Global and Separable Operations}
%The states under consideration are $\ket{\psi_i}=\ket{s_i}\otimes \ket{s_i}$ where
%\begin{align}
%\ket{s_0}&=\ket{0}\notag\\
%\ket{s_1}&=-\tfrac{1}{2}\ket{0}+\tfrac{\sqrt{3}}{2}\ket{1}\notag\\
%\ket{s_2}&=\tfrac{1}{2}\ket{0}+\tfrac{\sqrt{3}}{2}\ket{1}.
%\end{align}
The dual states of $\ket{\psi_i}=\ket{s_i}\otimes \ket{s_i}$ in (\ref{eq-doubletrine}) can be computed as
\begin{align}
\ket{\widetilde{\psi}_0}&=\tfrac{3}{\sqrt{10}}\ket{00}-\tfrac{1}{\sqrt{10}}\ket{11}\notag\\
\ket{\widetilde{\psi}_1}&=-\sqrt{\tfrac{3}{10}}(\ket{01}+\ket{10})+\sqrt{\tfrac{2}{5}}\ket{11}\notag\\
\ket{\widetilde{\psi}_2}&=\sqrt{\tfrac{3}{10}}(\ket{01}+\ket{10})+\sqrt{\tfrac{2}{5}}\ket{11}\notag.
\end{align}
Using the $V=U\otimes U$ symmetry, we can further simplify the problem.  Note that $\{\ket{\widetilde{\psi}_i}\}$ also demonstrate the symmetry $\ket{\widetilde{\psi}_i}=V^i\ket{\widetilde{\psi}_0}$.  Consequently, we have
\[
\frac{1}{3} \sum_{k=0}^2\bra{\psi_k}\Pi_k\ket{\psi_k} = \frac{1}{3}\sum_{k=0}^2\bra{\psi_k}\Pi_k+V\Pi_{k-1}V^\dagger+V^2\Pi_{k-2}(V^\dagger)^2\ket{\psi_k},
\]
and so we can replace any POVM $\{\Pi_0,\Pi_1,\Pi_2,\Pi_?\}$ by
\begin{align}
\label{Eq:POVMnew}
\hat\Pi_0&=1/3(\Pi_0+V\Pi_2V^\dagger+V^2\Pi_1 V^\dagger)\notag\\
\hat\Pi_1&=V\hat\Pi_0 V^\dagger\notag\\
\hat\Pi_2&=V^2\hat\Pi_0(V^\dagger)^2\notag\\
\hat\Pi_?&=1/3(\Pi_?+V\Pi_?V^\dagger+V^2\Pi_? (V^\dagger)^2\notag\\
&=\mathbb{I}-\Pi_0-V\hat\Pi_0 V^\dagger-V^2\hat\Pi_0(V^\dagger)^2.
\end{align}
Thus, $\bra{\psi_k}\hat\Pi_k\ket{\psi_k}$ is constant for all $k$, and so without loss of generality our problem is the following:
\begin{align}
\max_{\Pi\geq 0} \quad&\bra{\psi_0}\Pi\ket{\psi_0}\notag\\
\text{such that:}\quad& {{supp}}(\Pi)=\widetilde{S}_0\notag\\
&\sum_{k=0}^2V^k\Pi(V^\dagger)^k\leq\mathbb{I}.
\end{align}

For a separable POVM $\{\Pi_0,\Pi_1,\Pi_2,\Pi_?\}$, each $\Pi_k$ is a convex combination of rank one product projectors, and since $U\otimes U$ maps product states to product states, the modified POVM $\{\hat{\Pi}_0,\hat{\Pi}_1,\hat{\Pi}_2,\hat{\Pi}_?\}$ given by Eq. \eqref{Eq:POVMnew} will also be separable.  Hence, our new optimization problem is 
\begin{align}
\max_{\Pi\geq 0} \quad&\bra{\psi_0}\Pi\ket{\psi_0}\notag\\
\text{such that:}\quad& supp(\Pi)=\widetilde{S}_0\notag\\
&\sum_{k=0}^2V^k\Pi(V^\dagger)^k\leq\mathbb{I}\notag\\
&\text{$\Pi$ is separable}\notag\\
&\text{$\mathbb{I}-\sum_{k=0}^2V^k\Pi(V^\dagger)^k$ is separable}.
\end{align}
We take 
\begin{equation}
\label{Eq:POVMgeneral}
\Pi=a\op{\widetilde{\psi}_0}{\widetilde{\psi}_0}+b(\op{\widetilde{\psi}_0}{\Psi^-}+\op{\Psi^-}{\widetilde{\psi}_0})+c\op{\Psi^-}{\Psi^-}
\end{equation}
so that $\bra{\psi_0}\Pi\ket{\psi_0}=a$, the eigenvalues of $\Pi$ are
\begin{equation}
\label{Eq:POVM1}
\{\frac{1}{9}(5 a + 9 c \pm \sqrt{(5a-9c)^2 + 180 b^2})\}
\end{equation}
 and $\sum_{k=0}^2V^k\Pi(V^\dagger)^k$ has distinct eigenvalues of 
\begin{equation}
\label{Eq:POVM2}
 \{\frac{4}{3}a,\frac{1}{3}[a+9c\pm\sqrt{(a-9c)^2+36b^2}]\}.
\end{equation}
Putting aside the separability constraint, we thus see that the choice $a=3/4$ and $b=c=0$ is a feasible point which maximizes $\bra{\psi_0}\Pi\ket{\psi_0}$.  In other words, the optimal unambiguous probability for the double trine using global operations is $3/4$.

Now, to demand that $\Pi$ is separable, we compute its concurrence.  Recall that for a two-qubit mixed state $\rho$, its concurrence is given by $C(\rho)=\max\{0,\lambda_1^\downarrow-\lambda_2^\downarrow-\lambda_3^\downarrow-\lambda_4^\downarrow\}$ where the $\lambda_i^\downarrow$ are the square roots of the eigenvalues (in decreasing order) of the matrix $\rho\widetilde{\rho}$, where $\widetilde{\rho}=\sigma_y\otimes\sigma_y\rho^*\sigma_y\otimes\sigma_y$ \cite{Wootters-1998a}.  Without loss of generality, we can assume that $\Pi$ is real, and it will be separable if and only if its concurrence vanishes.  Since $\Pi$ is rank 2, this amounts to the two nonzero eigenvalues of $\Pi(\sigma_y\otimes \sigma_y)\Pi(\sigma_y\otimes\sigma_y)$ being equal.   Hence, we obtain the following constraint: 
\begin{align}
\label{Eq:SepCons1}
0&=(a - 3c)^2[(a+3c)^2-12b^2].
\end{align}
In addition, to this, we also need that $\Omega:=\mathbb{I}-\sum_{k=0}^2V^k\Pi(V^\dagger)^k$ is separable.  However, first let's focus on the optimization only under the constraint of \eqref{Eq:SepCons1}.  We are thus left with two cases: (i) $a=3c$ and (ii) $12b^2=(a+3c)^2$.  First consider case (ii).  Substituting into Eq. \eqref{Eq:POVM2}, the task is to maximize $a$ subject to $3\geq a+9c+2\sqrt{a^2+27c^2}$ and $a+9c-2\sqrt{a^2+27c^2}\geq 0$.  The maximum is obtained at the boundary, which is the point $a=3/8$, $b=\sqrt{3}/8$, and $c=1/8$.  On the other hand, for case (i) we maximize $a$ subject to $3\geq 4a+\sqrt{2a^2+36b^2}$ and $4a\geq\sqrt{2a^2+36b^2}$.  Again, optimality is obtained at the boundary, but this time with the point $a=1/2$, $b=0$, and $c=1/6$.

Now we turn to the operator $\Omega$.  The eigenvalues for $\Omega(\sigma_y\otimes \sigma_y)\Omega(\sigma_y\otimes\sigma_y)$ are 
\begin{align*}
t_1=&1-\frac{2}{3}a+\frac{2}{9} a^2 + 4 b^2 - 6 c + 18 c^2 \notag\\&+ 
 \frac{2}{9} (3 - a - 9 c)\sqrt{ (a-9c)^2 + 36 b^2}, \notag \\
 t_2=&1-\frac{2}{3}a+\frac{2}{9} a^2 + 4 b^2 - 6 c + 18 c^2 \notag \\ &-
 \frac{2}{9} (3 - a - 9 c)\sqrt{ (a-9c)^2 + 36 b^2},\\
t_3=&\frac{1}{9}(3-4a)^2,\notag\\ 
 t_4=&\frac{1}{9}(3-4a)^2.
\end{align*}
It can be verified that for the point $a=1/2$, $b=0$, and $c=1/6$ we have $t_1\to 4/9$, $t_2\to 0$, and $t_3=t_4\to 1/9$.  The concurrence of $\Omega$ is given by $\sqrt{t_1}-\sqrt{t_2}-\sqrt{t_3}-\sqrt{t_4}=0$.  Thus, the optimal point for separability of $\Pi$ is also a point in which $\Omega$ is separable.  So in summary, the optimal unambiguous probability for the double trine using separable operations is $1/2$.

\subsubsection*{LOCC Operations}
We next describe an LOCC protocol that also obtains an unambiguous probability of 1/2.  It is, in fact, the one described in Section \ref{Sect:Ncopy}.  Consider the states
\begin{align*}
\ket{\overline{s_0}}&=\ket{1}\notag\\
\ket{\overline{s_1}}&=\frac{\sqrt{3}}{2}\ket{0}+\frac{1}{2}\ket{1}\notag\\
\ket{\overline{s_2}}&=\frac{\sqrt{3}}{2}\ket{0}-\frac{1}{2}\ket{1}.
\end{align*}
Note that $|\ip{\widetilde{s_i}}{s_j}|=\sqrt{3}/2$ if $i\not=j$ and $0$ if $i=j$.  It can be verified that the set 
\[
\mathcal{P}=\left\{\frac{2}{3}\op{\overline{s_0}}{\overline{s_0}},\frac{2}{3}\op{\overline{s_1}}{\overline{s_1}},\frac{2}{3}\op{\overline{s_2}}{\overline{s_2}}\right\}
\] constitutes a valid POVM.  The protocol consists of Alice and Bob each performing the POVM $\mathcal{P}$ and comparing their results.  If they obtain different outcomes, then they know the state they share is the one distinct from each of their outcomes.  For instance if Alice obtains $\ket{\overline{s_0}}$ and Bob obtains $\ket{\overline{s_2}}$, then they can conclusively deduce that their state is $\ket{\psi_1}=\ket{s_1}\otimes\ket{s_1}$.  Thus, the only time they cannot determine their state is when they both obtain the same outcome.  This occurs with probability:
\begin{align}
\frac{1}{3}\cdot\sum_{i\not=j}\left(\frac{2}{3}\right)^2|\ip{\overline{s_i}}{s_j}|^4=\left(\frac{2}{3}\right)^2\cdot 2\cdot \frac{9}{16}=1/2.
\end{align}
Therefore, the optimal probability of unambiguous discrimination via LOCC is 1/2.

\section{Conclusion}
In this paper, we have provided conditions under which various ensembles of two-qubit states can either be perfectly or optimally distinguished by LOCC. These results significantly advance the current understanding of state discrimination for two-qubit ensembles. For perfect LOCC discrimination, we provide new instances of necessary and sufficient conditions that are much easier to verify than the condition given in \cite{Duan-2010a}. Additionally, we have provided a necessary and sufficient condition for which the two-qubit ensembles consisting of one pure state and one rank two mixed state can be perfectly distinguished by separable operations; thus completing the previously missing piece in the perfect distinguishability setting. With this, perfect discrimination of two-qubit ensembles by both LOCC and SEP operations is completely solved. 

Most notably, we have observed sharp distinctions between ensembles consisting of two states and those consisting of three states. First, we have shown that \emph{almost all} two-qubit ensembles consisting of three pure states cannot be optimally discriminated using LOCC; in contrast, \emph{any} two pure states can be optimally distinguished by LOCC \cite{Walgate-2000a}. Furthermore, we have demonstrated that the $N$-copy trine ensemble cannot be optimally distinguished by LOCC for any finite $N$. Again, this behavior is the complete opposite than if there were only two $N$-copy states, which can be optimally distinguished by LOCC  \cite{Acin-2005a}. 

We would like to emphsize the interesting connection between the $N$-copy trine ensemble for $N\geq3$ and Shor's lifted trine ensemble, where each positive integer $N$ corresponds to a certain lifting angle \cite{Shor-2002a}. This observation allows us to simplify the computation by mapping the $N$-qubit trine states of higher dimensions into a three-dimensional subspace in $\mathbb{R}^3$. 

Finally, we have also observed very bizarre distinguishability features for the double trine ensemble; namely, we have shown in this paper that for optimal unambiguous discrimination: LOCC$=$SEP$<$GLOBAL. This finding is rather different than a previously obtained result that LOCC$<$SEP$=$GLOBAL when minimum error discrimination is considered \cite{Chitambar-2013b}.  This raises the intriguing question of whether there exists certain ensembles for which LOCC$<$GLOBAL with respect to one performance measure but LOCC$=$GLOBAL with respect to another.  If the answer is positive, then the phenomenon of nonlocality without entanglement might not be a property that depends solely on the underlying states themselves.

\section*{Acknowledgement}
The authors would like to thank Debbie Leung and Laura Mancinska for helpful discussions on the topic of LOCC distinguishability. RD was supported in part by the Australian Research Council (ARC) under Grant DP120103776 (with A.Winter) and by the National Natural Science Foundation of China under Grants 61179030. He was also supported in part by an ARC Future Fellowship under Grant FT120100449. MH was supported by the Chancellor's postdoctoral research fellowship, University of Technology Sydney.

\bibliographystyle{IEEEtran}
\bibliography{QuantumBib}

\end{document}